\documentclass[11pt,fleqn]{article}

\usepackage[margin=1in]{geometry}
\usepackage[authoryear,longnamesfirst]{natbib}
\usepackage{amsmath,amssymb}
\usepackage{amsthm}
\usepackage{header}
\usepackage[colorlinks=true,linkcolor=blue,citecolor=blue,urlcolor=blue]{hyperref}

\newenvironment{pf}[1][]{%
  \if\relax\detokenize{#1}\relax
    \begin{proof}%
  \else
    \begin{proof}[Proof (#1)]%
  \fi
}{\end{proof}}

\usepackage{pgfplots}
\pgfplotsset{compat=1.18}
\definecolor{cbblue}{RGB}{0,114,178}
\definecolor{cbverm}{RGB}{213,94,0}
\definecolor{cbgreen}{RGB}{0,158,115}
\definecolor{cbpurple}{RGB}{106,61,154}
\pgfplotsset{
  maxmin axis/.style={
    width=\linewidth, height=0.82\linewidth,
    xmin=-0.02, xmax=1.02, ymin=0, ymax=0.66,
    axis lines=left, line cap=round,
    xlabel={$p$},
    tick label style={font=\footnotesize},
    label style={font=\footnotesize},
    every axis plot/.append style={line join=round},
  }
}

\title{Zero-Sum Fictitious Play Cannot Converge to a Point}
\author{Jaehong Moon\thanks{Department of Industrial and Enterprise
  Systems Engineering, University of Illinois Urbana-Champaign, Urbana,
  IL, USA. Email: \texttt{jm133@illinois.edu}.}}
\date{}

\begin{document}
\maketitle

\begin{abstract}
Fictitious play is a history-based learning process in which players
best respond to the empirical distribution of their opponents' past play.
Although classical results show that fictitious play converges to the set
of Nash equilibria in zero-sum games, this set convergence does not imply
convergence to a single equilibrium when the equilibrium set is non-singleton.
This paper shows that pointwise convergence fails in a strong sense.
Suppose the equilibrium set of a player has positive measure and
consists only of fully mixed strategies. Then, under any tie-breaking
rule, the empirical strategy of that player cannot converge to any
equilibrium point, provided it is initialized outside the equilibrium
set---in particular, when the player starts with no prior beliefs.
The proof identifies two mechanisms behind this instability.
In the interior of the equilibrium set, the dynamics retain inertia
that prevents settling. At the boundary, the opponent's deviations
from its unique equilibrium action steadily unbalance the cumulative
payoffs that drive best responses, so that not all actions can remain
competitive, as a fully mixed limit would require. The results clarify the gap between convergence to an
equilibrium set and convergence to an equilibrium point in fictitious play.
\end{abstract}


\section{Introduction}
Fictitious play (FP), introduced by \citet{brown_fp}, is a natural learning process for finding a Nash equilibrium (NE) in a strategic normal-form game. In FP, each player is myopic and chooses its action based solely on the observed history of the opponents’ actions. Since each player only needs to keep track of the payoffs of its own actions at each iteration, FP provides a fundamental and flexible learning procedure.

A line of research has studied which classes of games admit convergence of FP to a NE. Such classes are said to have the fictitious play property (FPP). Notably, FP converges in zero-sum games~\citep{Robinson}, \(2\times 2\) games~\citep{miyasawa,2by2_geom}, potential games~\citep{potential}, games with identical interests~\citep{identical}, and (nondegenerate) \(2\times n\) games~\citep{2byn_fpp}. However, not every game has the FPP; a well-known counterexample was given by \citet{shapley_game}.

A closer look reveals that the FPP may depend on the specific tie-breaking rule. \citet{miyasawa} showed that \(2\times 2\) games have the FPP under a particular tie-breaking rule. However, it was later shown by \citet{2by2_without} that, under a different tie-breaking rule, FP may fail to converge to a NE in some \(2\times 2\) games. The FPP for \(2\times 2\) games independently of the tie-breaking rule was subsequently established by \citet{identical,2byn_fpp}, under an additional nondegeneracy assumption.

Although the primary focus of FP research is whether FP converges to a NE, a natural question arises when multiple NE exist: which equilibrium does FP learn? Does FP learn different equilibria under different tie-breaking rules? If so, can we characterize the properties of the equilibria learned under the corresponding tie-breaking rules? The zero-sum setting is a natural starting point for these questions: recall that convergence to the NE set is guaranteed~\citep{Robinson}. Moreover, the NE set of a zero-sum game is closed, convex, and hence connected, through its connection to linear programming~\citep{adler,dantzig1951proof}, so it may contain a continuum of equilibria among which FP must implicitly select. Motivated by these observations, in this work we investigate the pointwise convergence of FP dynamics in zero-sum games.

The main result of this paper is that, in certain zero-sum games, FP cannot converge to a single point under any sequence of best responses. Our non-convergence result is based on two conditions: (Assumption~\ref{asspt:positive_measure}) the NE set for one player has positive measure, and (Assumption~\ref{asspt:boundary}) every NE in that set is fully mixed; in addition, we normalize the game by removing strictly dominated actions (Assumption~\ref{asspt:no_dominated}) and require the player's initial empirical strategy to lie outside its equilibrium set, as is the case when the player starts with no prior beliefs. At a high level, the proof treats interior and boundary equilibria separately. In the interior of the NE set, the player's dynamics exhibit inertia (Lemma~\ref{lemma:inertia}), causing the trajectory to escape the NE set rather than settle at a single equilibrium. At boundary equilibria, we show that the opponent's deviations from its unique equilibrium action never reward all of the player's actions equally (Lemma~\ref{lemma:no_constant}); consequently, the cumulative payoffs that govern the player's best responses develop an imbalance proportional to the number of deviations, and a counting argument shows that the player cannot keep all of its actions competitive while its empirical strategy remains in a small neighborhood of a fixed boundary equilibrium (Proposition~\ref{prop:boundary_instability}).

The paper is organized as follows. In Section~\ref{sec:tie_non_conv}, we begin with simple examples showing that tie-breaking can strongly influence the pointwise behavior of fictitious play. In Section~\ref{sec:motivation}, we construct zero-sum games with nontrivial Nash equilibrium sets of positive measure and present examples suggesting tie-breaking-independent non-convergence. Section~\ref{sec:main_result} introduces Assumptions~\ref{asspt:positive_measure} and~\ref{asspt:boundary} and states our main non-convergence theorem. The proof is given in Section~\ref{sec:proof}: interior and boundary points of the equilibrium set are handled by an inertia argument and a counting argument, respectively. Finally, Section~\ref{sec:discussion} examines examples related to the necessity of Assumptions~\ref{asspt:positive_measure} and~\ref{asspt:boundary} and the role of initialization, and discusses the implications of our results for observable stopping criteria in learning.

\textbf{Notation}. 
Throughout the paper, lowercase letters denote scalars, and boldface letters denote vectors, and $A$ represents a matrix. When referring to a concrete example, we place a tilde on the corresponding notation (e.g., $\tilde{A}$ for a particular matrix instance).

Let \([n] := \{1,2,\ldots,n\}\) for \(n \in \NN\).
We define the all-one vector as \(\onevector_n = (1,\dots,1)^\top \in \RR^n\) and the all-zero vector as \(\zerovector_n = (0,\dots,0)^\top \in \RR^n\), and denote the standard basis of \(\RR^n\) by \(\basis_n=\{\bfe_1,\ldots,\bfe_n\}\).
Whenever the dimension \(n\) is clear from context, we simply write \(\bfe_i\) without explicitly specifying that \(\bfe_i \in \RR^n\).
For a matrix \(A\in \RR^{n\times m}\) and \(i \in [m]\), let \(A_{-i}\) be the \(n \times (m-1)\) matrix obtained by removing the \(i\)-th column of \(A\).
For a vector \(\bfw \in \RR^k\) and \(i \in [k]\), let \([\bfw]_i\) be the \(i\)-th component of \(\bfw\) and \((\bfw)_{-i}\in \RR^{k-1}\) be the vector obtained by removing the \(i\)-th component of \(\bfw\).

Let
$S_n=\{\bfx\in \RR^n : x_i\ge 0 \text{ for all } i\in[n],\ \sum_{i=1}^n x_i =1\}$
denote the probability simplex in $\RR^n$. We regard $S_n$ as a subset of its affine hull $\operatorname{aff}(S_n)$, equipped with the relative topology, and identify $\operatorname{aff}(S_n)$ with $\RR^{n-1}$ via the homeomorphism $(x_1,\dots,x_n)\mapsto(x_1,\dots,x_{n-1})$. For $U\subseteq \operatorname{aff}(S_n)$, we write $\intr(U)$ and $\partial U$ for its interior and boundary in this topology; when $U\subseteq S_n$, these are the relative interior and relative boundary of $U$. Likewise, we let $\mu$ denote the measure on $S_n$ obtained from Lebesgue measure on $\RR^{n-1}$ under this identification.

\section{Preliminaries}
\subsection{Zero-sum Game}
A two-player normal-form game $(\{\Ac_1,\Ac_2\},\{u_1,u_2\})$ with action sets $\Ac_1=[n]$ and $\Ac_2=[m]$ is called a zero-sum game if its payoff functions satisfy $u_1(i,j)=-u_2(i,j)=A_{ij}$ for some matrix $A \in \RR^{n \times m}$. Such a game is fully described by the matrix $A$, which is referred to as the payoff matrix of the game. We mainly consider the mixed extension of the game, in which a mixed strategy for each player is a probability distribution over its actions.
We represent Player~1's mixed strategy by ${\bfx\in S_n}$ and Player~2's mixed strategy by $\bfy\in S_m$. Given mixed strategies $(\bfx,\bfy)$, the payoffs of Player~1 and Player~2 are $\bfx^\top A \bfy$ and $-\bfx^\top A \bfy$, respectively. 

For a payoff matrix $A \in \RR^{n \times m}$, let $\bfc_i$ and $\bfr_j$ denote its $i$-th column and $j$-th row, respectively, i.e.,
\begin{align}\label{eq:mat_col_row}
     A \;=\; 
     \begin{pmatrix}
         \vert &        & \vert \\
         \bfc_1   & \cdots & \bfc_m   \\
         \vert &        & \vert
     \end{pmatrix}
     \;=\;
     \begin{pmatrix}
         \text{---}\hspace{-0.2cm} & \bfr_1^\top & \hspace{-0.2cm}\text{---} \\
                                   & \vdots   &                          \\
         \text{---}\hspace{-0.2cm} & \bfr_n^\top & \hspace{-0.2cm}\text{---}
     \end{pmatrix}.
\end{align}

In the zero-sum setting, we call a tuple $(\bfx^*, \bfy^*) \in S_n\times S_m$ a (mixed) Nash equilibrium (NE) if $(\bfx^*, \bfy^*)$ is a saddle point of $\bfx^\top A \bfy$, i.e.,
\begin{align*}
    {\bfx^*}^\top A \bfy^* \ge \bfx^\top A \bfy^* \quad \forall \bfx \in S_n, \quad \text{and} \quad
    {\bfx^*}^\top A \bfy^* \le {\bfx^*}^\top A \bfy \quad \forall \bfy \in S_m.
\end{align*}
In particular, $(\bfx^*, \bfy^*)$ is a NE if and only if $\bfx^*$ and $\bfy^*$ solve
\begin{align}
    &\max_{\bfx \in S_n}\;\min_{i \in [m]} \;\bfc_i^\top \bfx,\quad\text{and }\label{eq:lp_duality_x}\\
    &\min_{\bfy \in S_m}\;\max_{j \in [n]} \;\bfr_j^\top \bfy,\label{eq:lp_duality_y}
\end{align}
respectively.

Since the two problems \eqref{eq:lp_duality_x} and \eqref{eq:lp_duality_y} are independent, we define the NE set of Player~1 by $\nex$ and that of Player~2 by $\ney$. Since the above problems are convex programs over compact domains, the solution sets $\nex$ and $\ney$ are closed and convex. By von Neumann's minimax theorem~\citep{min_max}, for every $\bfx^* \in \nex$ and $\bfy^* \in \ney$, the pair $(\bfx^*, \bfy^*)$ is a NE. Hence, if we denote the set of all NE of the game by $\NE$, we can write
\begin{align*}
    \NE=\nex\times \ney.
\end{align*}
In addition, for any $(\bfx^*, \bfy^*)\in \nex \times \ney$, the payoff ${\bfx^*}^\top A \bfy^*$ is equal to the value $v$ of the game, where
\begin{align*}
    v:=\min_{\bfy\in S_m}\max_{\bfx\in S_n}\bfx^\top A\bfy
    =\max_{\bfx\in S_n}\min_{\bfy\in S_m}\bfx^\top A\bfy.
\end{align*}


\subsection{Fictitious Play (FP)}\label{sec:fp_set_up}

For a two-player game, FP is an iterative process in which each player selects a best response to the empirical frequency of the opponent's past actions. We define the best-response correspondences of the two players as set-valued maps from the opponent's mixed strategy to the set of pure actions. More precisely, in the zero-sum setting, we define the best-response correspondences of the row and column players by $\brx:S_m \cup \{\zerovector_m\}\rightrightarrows [n]$ and $\bry:S_n\cup \{\zerovector_n\}\rightrightarrows [m]$, respectively, where
\begin{align*}
    \brx(\bfy)
    :&= \argmax_{i\in [n]} \bfr_i^\top \bfy \\
    &= \bigl\{ i \in [n] : \bfe_i^\top A \bfy \ge \tilde \bfx^\top A \bfy \quad \forall \tilde \bfx \in \basis_n \bigr\},\\
    \bry(\bfx)
    :&= \argmin_{j\in [m]} \bfc_j^\top \bfx \\
    &= \bigl\{ j \in [m] : \bfx^\top A \bfe_j \le \bfx^\top A \tilde \bfy \quad \forall \tilde \bfy \in \basis_m \bigr\}.
\end{align*}
Note that we extend the best response to $\zerovector$: $\brx(\zerovector_m)=[n]$ and $\bry(\zerovector_n)=[m]$. Here, the notation $F:X\rightrightarrows Y$ denotes a set-valued map, that is, a map assigning to each $x\in X$ a subset $F(x)\subseteq Y$.

Central to our analysis are the (closed) best-response polyhedral regions defined by
\begin{align*}
    X_i := \bigl\{\bfx \in S_n : \bfc_i^\top \bfx \le \bfc_k^\top \bfx,\; \forall k \in [m]\bigr\},\quad \text{and}\quad
    Y_j := \bigl\{\bfy \in S_m : \bfr_j^\top \bfy \ge \bfr_k^\top \bfy,\; \forall k \in [n]\bigr\},
\end{align*}
for all $i\in [m]$ and $j\in [n]$. 
Note that if $\bfx\in X_i$, then Action~$i$ is among the best responses of Player~2, that is, $i\in \bry(\bfx)$. 
Similarly, $\bfy\in Y_j$ implies that Action~$j$ is among the best responses of Player~1.

We define FP (or FP dynamics) as a sequence $\{(\hat\bfx(k), \hat\bfy(k))\}_{k\ge 0}$ evolving according to
\begin{align}\label{eq:dynamics}
    \begin{split}
    \hat\bfx(k + 1)  = \frac{(k + k_1)\,\hat\bfx(k)\;+\;\bfe_{p_k}}{k + k_1 + 1},\quad
    \hat\bfy(k + 1)  = \frac{(k + k_2)\,\hat\bfy(k)\;+\;\bfe_{q_k}}{k + k_2 + 1},
    \end{split}
\end{align}
where $p_k \in \brx\bigl(\hat\bfy(k)\bigr)$ and $q_k \in \bry\bigl(\hat\bfx(k)\bigr)$. The process is initialized by parameters $k_1,k_2 \ge 0$ and vectors $\hat\bfx(0),\hat\bfy(0)$, where $\hat\bfx(0)\in S_n$ if $k_1>0$ and $\hat\bfx(0)=\mathbf{0}$ if $k_1=0$, and similarly $\hat\bfy(0)\in S_m$ if $k_2>0$ and $\hat\bfy(0)=\mathbf{0}$ if $k_2=0$. Intuitively, $\hat\bfx(0)$ and $\hat\bfy(0)$ represent prior beliefs about the opponent’s strategy, while $k_1$ and $k_2$ quantify each player’s confidence in these priors.

If multiple best responses arise at some step, each player selects one of them according to a tie-breaking rule. We allow arbitrary tie-breaking rules, possibly depending on $\hat\bfx(k)$, $\hat\bfy(k)$, $k$, randomness, or additional state variables, as long as the selected action remains in the corresponding best-response set. For $k\ge 1$, the vectors $\hat\bfx(k)$ and $\hat\bfy(k)$ can be interpreted as the empirical frequencies of actions played by each player up to step $k$.

\citet{Robinson} showed that, for any zero-sum game, under any tie-breaking rule and from any admissible initial condition, the corresponding FP dynamics $\{(\hat\bfx(k), \hat\bfy(k))\}$ converge to the product set $\NE=\nex\times\ney$. More precisely, for any $\epsilon>0$, there exists $K_\ep \in \NN$ such that
\begin{align*}
    \dist{(\hat\bfx(k), \hat\bfy(k))}{\NE}<\epsilon
\end{align*}
for all $k>K_\ep$, where
\begin{align*}
    \dist{(\bfx,\bfy)}{\NE}:=\inf_{(\bfx^*,\bfy^*)\in \nex\times\ney}\norm{(\bfx,\bfy)-(\bfx^*,\bfy^*)}_2.
\end{align*}

\section{Tie-Breaking Can Cause Non-Convergence}\label{sec:tie_non_conv}
It is clear that convergence to the NE set does not imply convergence of the dynamics to a single limit point. To construct a non-convergent FP trajectory, it is natural to consider modifying the tie-breaking rule. As a simple illustration, suppose there exist $i,j \in [n]$ such that the $i$-th and $j$-th rows of the payoff matrix $A$ are identical. Then $i\in \brx\bigl(\hat \bfy(k)\bigr)$ if and only if $j\in \brx\bigl(\hat \bfy(k)\bigr)$ for every $k$. This suggests that, under a suitable tie-breaking rule, one may induce oscillatory behavior in $[\hat\bfx(k)]_i$ and $[\hat\bfx(k)]_j$.

To isolate the essential issue, consider the case in which Player~1 has two actions, say Actions~$i$ and~$j$, that are always tied as best responses. If we track only the relative frequencies with which these two actions are played, the induced dynamics reduce to FP for the zero-sum game with payoff matrix
\begin{align}\label{eq:zz_mat}
    \tilde A = \begin{pmatrix} 0 \\ 0 \end{pmatrix}.
\end{align}
This is a completely non-strategic game, since every strategy profile is a NE. Thus, any non-convergence must come entirely from the tie-breaking rule rather than from the strategic structure of the game itself. Under many familiar tie-breaking rules, however, FP for the game $\tilde A$ does converge to a point. For example, if the tie-breaking rule always selects the best response with the lowest index, the iteration converges to $(1,0)^\top$. If, instead, one selects among multiple best responses uniformly at random, then the iteration converges almost surely to $(1/2,1/2)^\top$. This leads to the main question of this section: can simple tie-breaking rules already generate non-convergent FP behavior for the game with payoff matrix~\eqref{eq:zz_mat}?

Suppose Player~1 uses a one-bit memory state $m(k)$. Fix $0<a<b<1$, and initialize $m(0)=1$. We update $m(k)$ according to
\begin{align*}
    m(k+1)=
    \begin{cases}
        0 & \text{if } [\hat\bfx(k+1)]_2\ge b,\\
        1 & \text{if } [\hat\bfx(k+1)]_2\le a,\\
        m(k) & \text{otherwise}.
    \end{cases}
\end{align*}
We then define the tie-breaking rule for Player~1 by
\begin{align*}
    \tau_x(m)=
    \begin{cases}
        1 & \text{if } m=0,\\
        2 & \text{if } m=1.
    \end{cases}
\end{align*}
It is straightforward to verify that, under this scheme, $[\hat\bfx(k)]_2$ oscillates between $a$ and $b$ rather than converging to a single point.

Interestingly, even without any auxiliary one-bit memory, oscillatory behavior can still arise. Consider the tie-breaking rule $\tau_x:S_2\to [2]$ defined by
\begin{align}\label{eq:no_memory_oscillate}
    \tau_x(\bfx)
    \;=\;
    \begin{cases}
        2 & \text{if } [\bfx]_2 = 0,\\
        1 & \text{if } [\bfx]_2 = 1,\\
        2 & \text{if } [\bfx]_2 = \frac{p}{q}\text{ with }(p,q)=1\text{ and }p\text{ odd},\\
        1 & \text{if } [\bfx]_2 = \frac{p}{q}\text{ with }(p,q)=1\text{ and }p\text{ even}.
    \end{cases}
\end{align}

\begin{proposition}\label{prop:oscillation}
Starting from $\hat\bfx(0)=\mathbf{0}$, the FP trajectory for the non-strategic game with payoff matrix~\eqref{eq:zz_mat} and tie-breaking rule~\eqref{eq:no_memory_oscillate} visits the points $p_1=(1/2,1/2)^\top$ and $p_2=(1/3,2/3)^\top$ infinitely often. More precisely, for every $n\ge 1$,
\begin{align}
    [\hat\bfx(2^n)]_2 &= \frac{1}{2},\label{eq:odd_case}\\
    [\hat\bfx(2^n+2^{n-1})]_2 &= \frac{2}{3}.\label{eq:even_case}
\end{align}
In particular, the trajectory does not converge.
\end{proposition}
The proof proceeds by induction on $n$ and is deferred to Appendix~\ref{app:proof_oscillation}.

The examples in this section show that tie-breaking is not a mere technical detail, but can fundamentally affect the pointwise behavior of FP. This naturally raises a broader question: to what extent can tie-breaking control pointwise convergence? Can an appropriate tie-breaking rule enforce convergence, or are there games in which non-convergence is unavoidable regardless of how ties are resolved? The rest of the paper addresses this question by showing that there exist examples for which tie-breaking cannot rescue pointwise convergence.

\section{Positive Measure NE Sets and How to Find Them}\label{sec:motivation}
In the previous section, we examined non-convergence in a game with a trivial NE set, namely, one in which every strategy profile is a NE. A natural next question is whether a zero-sum game can have a non-trivial and non-singleton NE set and, if so, how such games can be constructed. It is well known that if the entries of the payoff matrix $A$ are sampled from a continuous distribution, then the resulting zero-sum game admits a unique NE~\citep{van1991stability,daskalakis2019last}. In many applications, however, the entries of a normal-form game come from a discrete or otherwise structured set, such as the integers $\mathbb{Z}$, which leaves open the possibility of multiple equilibria. In this section, we present an idea to construct zero-sum games with non-trivial NE sets of positive measure.

\subsection{Study of 2-Action Games}
We begin with the simplest case in which Player~1 has only two actions. To characterize Player~1's equilibrium strategies, we study problem~\eqref{eq:lp_duality_x}. Let $A \in \RR^{2 \times m}$ and write $\bfx = (p,1-p)^\top \in S_2$. Then problem~\eqref{eq:lp_duality_x} becomes
\begin{align*}
    \max_{\bfx \in S_2}\;\min_{i \in [m]} \bfc_i^\top \bfx
    =
    \max_{p \in [0,1]}\;\min_{i \in [m]} \bigl\{a_{1i}p + a_{2i}(1-p)\bigr\}.
\end{align*}
Graphically, this corresponds to plotting the lower envelope of the affine functions $a_{1i}p + a_{2i}(1-p)$ over $p \in [0,1]$ and identifying the maximizer.

For example, consider the game
\begin{align}\label{eq:2by2_basic}
    \tilde A = \begin{pmatrix}
        1 & 0 \\ 0 & 1
    \end{pmatrix}.
\end{align}
Figure~\ref{fig:2by2_ex_basic} illustrates the graph of the max-min problem, from which we obtain $\tildenex = \{(1/2,1/2)^\top\}$. To obtain multiple equilibria in this example, we would like to add another affine function to the plot that is constant in $p$ and lies below the value $1/2$, so that it becomes the binding part of the lower envelope. Equivalently, we add a column $(\tilde a_{10},\tilde a_{20})^\top$ such that $\tilde a_{10}p + \tilde a_{20}(1-p)$ is constant and strictly less than $1/2$. That constant then becomes the value of the game. For instance, adding $(1/4,1/4)^\top$ yields
\begin{align}\label{eq:2by2_mult_ne}
    \tilde A = \begin{pmatrix}
         1 & 0 & 1/4 \\ 0 & 1 & 1/4
    \end{pmatrix}.
\end{align}
Its max-min graph is shown in Figure~\ref{fig:2by2_ex_multi_ne}, and it yields the nontrivial equilibrium set
\begin{align*}
    \tildenex = \{\theta \bfe_1 + (1-\theta)\bfe_2 : \theta \in [1/4,3/4]\}.
\end{align*}

\begin{figure}
  \centering
  \begin{subfigure}[t]{0.48\linewidth}
    \centering
    \begin{tikzpicture}
      \begin{axis}[maxmin axis,
          ylabel={payoff},
          xtick={0,0.5,1}, xticklabels={$0$,$\tfrac12$,$1$},
          ytick={0,0.25,0.5}, yticklabels={$0$,$\tfrac14$,$\tfrac12$}]
        \addplot[cbblue, semithick, domain=0:1, samples=2] {x};
        \addplot[cbverm, semithick, domain=0:1, samples=2] {1-x};
        \addplot[black, very thick] coordinates {(0,0) (0.5,0.5) (1,0)};
        \addplot[black, dashed, thin] coordinates {(0.5,0) (0.5,0.5)};
        \addplot[only marks, mark=*, mark size=1.5pt] coordinates {(0.5,0.5)};
        \node[font=\footnotesize, anchor=south east, cbverm] at (axis cs:0.42,0.55) {$\bfc_2^\top\bfx$};
        \node[font=\footnotesize, anchor=south west, cbblue] at (axis cs:0.58,0.55) {$\bfc_1^\top\bfx$};
        \node[font=\footnotesize, anchor=west] at (axis cs:0.54,0.47) {$v=\tfrac12$};
      \end{axis}
    \end{tikzpicture}
    \caption{The game in \eqref{eq:2by2_basic}.}
    \label{fig:2by2_ex_basic}
  \end{subfigure}\hfill
  \begin{subfigure}[t]{0.48\linewidth}
    \centering
    \begin{tikzpicture}
      \begin{axis}[maxmin axis,
          xtick={0,0.25,0.75,1}, xticklabels={$0$,$\tfrac14$,$\tfrac34$,$1$},
          ytick={0,0.25,0.5}, yticklabels={$0$,$\tfrac14$,$\tfrac12$}]
        \addplot[cbblue, semithick, domain=0:1, samples=2] {x};
        \addplot[cbverm, semithick, domain=0:1, samples=2] {1-x};
        \addplot[cbgreen, semithick, domain=0:1, samples=2] {0.25};
        \addplot[black, very thick] coordinates {(0,0) (0.25,0.25) (0.75,0.25) (1,0)};
        \addplot[black, dashed, thin] coordinates {(0.25,0) (0.25,0.25)};
        \addplot[black, dashed, thin] coordinates {(0.75,0) (0.75,0.25)};
        \addplot[cbpurple, line width=2.2pt] coordinates {(0.25,0.006) (0.75,0.006)};
        \node[font=\footnotesize, anchor=south east, cbverm] at (axis cs:0.42,0.55) {$\bfc_2^\top\bfx$};
        \node[font=\footnotesize, anchor=south west, cbblue] at (axis cs:0.58,0.55) {$\bfc_1^\top\bfx$};
        \node[font=\footnotesize, anchor=south, cbgreen] at (axis cs:0.5,0.26) {$\bfc_3^\top\bfx$};
        \node[font=\footnotesize, anchor=south, cbpurple] at (axis cs:0.5,0.03) {$\tildenex$};
      \end{axis}
    \end{tikzpicture}
    \caption{The game in \eqref{eq:2by2_mult_ne}.}
    \label{fig:2by2_ex_multi_ne}
  \end{subfigure}
  \caption{Max--min graphs. In each panel, the thin colored lines are the column payoffs $\bfc_i^\top\bfx$ with $\bfx=(p,1-p)^\top$, and the thick black line is the lower envelope $\min_{i\in[m]}\bfc_i^\top\bfx$; its maximizers form Player~1's equilibrium set. (a)~For the game in \eqref{eq:2by2_basic}, the maximizer is unique and the value is $1/2$. (b)~For the game in \eqref{eq:2by2_mult_ne}, the value is $1/4$ and every $p\in[1/4,3/4]$ is optimal, so the equilibrium set $\tildenex=\{\theta\bfe_1+(1-\theta)\bfe_2 : \theta\in[1/4,3/4]\}$ (purple segment) has positive measure.}
  \label{fig:maxmin_pair}
\end{figure}

Generally, the above construction can be summarized as follows. Let $v$ be the value of a zero-sum game with payoff matrix $A \in \RR^{n\times m}$. By inserting a column of the form $v'\onevector_n$ with $v' < v$ into $A$, we can produce a game whose equilibrium set is non-singleton. In fact, this construction yields an equilibrium set of positive measure. On the other hand, by adding a suitably chosen column, one can also obtain an equilibrium set with multiple points but measure zero; cf. the game $\tilde A$ in Equation~\eqref{eq:converging_example} of Section~\ref{sec:discussion}.

\subsection{Positive-Measure NE Sets May Not Guarantee Pointwise Convergence}
Having constructed a game with a positive-measure NE set, it is natural to apply FP to this game and investigate its behavior. Running FP for the game $\tilde A$ in \eqref{eq:2by2_mult_ne} with no prior yields the $10^4$-step trajectory shown in Figure~\ref{fig:2by2_fp}. Whenever multiple best responses arise, we choose one of them uniformly at random. As the figure illustrates, the FP trajectory does not converge to a single point.

More importantly, this non-convergence is not an artifact of the particular tie-breaking rule used in the simulation. In fact, one can show directly that FP for the game \eqref{eq:2by2_mult_ne} cannot converge to a single point under any tie-breaking rule. Equivalently, as long as each player continues to select actions from the corresponding best-response set, the FP trajectory fails to stabilize at any fixed strategy profile. 

\begin{figure}
  \centering
  \begin{subfigure}[b]{0.46\linewidth}
    \centering
    \includegraphics[width=\linewidth]{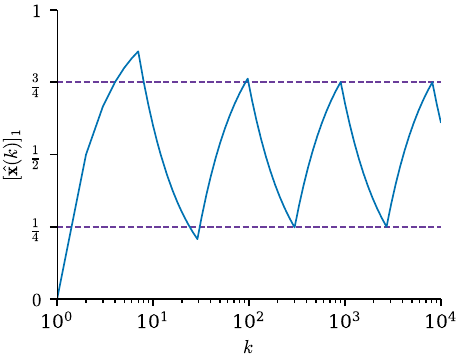}
    \caption{The $2\times3$ game \eqref{eq:2by2_mult_ne}.}
    \label{fig:2by2_fp}
  \end{subfigure}\hfill
  \begin{subfigure}[b]{0.52\linewidth}
    \centering
    \includegraphics[width=\linewidth]{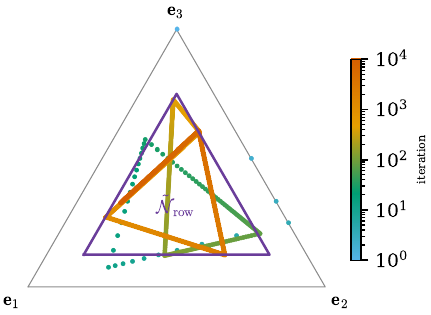}
    \caption{The $3\times4$ game \eqref{eq:non_converge_example}.}
    \label{fig:fp_example}
  \end{subfigure}
  \caption{Sample FP trajectories, starting from no prior and breaking ties uniformly at random among multiple best responses. (a)~The first coordinate $[\hat\bfx(k)]_1$ for the game in \eqref{eq:2by2_mult_ne} exhibits persistent oscillations between $1/4$ and $3/4$ (dashed lines), the extreme points of the equilibrium set; the horizontal axis is logarithmic in $k$. (b)~The empirical strategy $\hat\bfx(k)$ for the game in \eqref{eq:non_converge_example}, plotted in the simplex $S_3$ with color indicating the iteration. The purple triangle is the equilibrium set $\tildenex$ of \eqref{eq:non_converge_example_ne}: the trajectory approaches $\tildenex$ but keeps wandering inside it rather than settling at a point.}
  \label{fig:fp_traj_pair}
\end{figure}

The same construction extends naturally to the $3\times 4$ game
\begin{align}
\tilde A =
\begin{pmatrix}
    1/8 & 1 & 0 & 0 \\
    1/8 & 0 & 1 & 0 \\
    1/8 & 0 & 0 & 1
\end{pmatrix}.
\label{eq:non_converge_example}
\end{align}
The equilibrium sets of this game are given by
\begin{align}
\tildenex = \conv\bigl(\{\nae_1,\nae_2,\nae_3\}\bigr),
\qquad
\tildeney = \{\bfe_1\},
\label{eq:non_converge_example_ne}
\end{align}
where $\nae_1,\nae_2,\nae_3$ are the extreme points of $\tildenex$ given by
\begin{align*}
    \nae_1 &=
    \begin{pmatrix}
        \frac{3}{4}\\[0.1cm]
        \frac{1}{8}\\[0.1cm]
        \frac{1}{8}
    \end{pmatrix},
    \qquad
    \nae_2 =
    \begin{pmatrix}
        \frac{1}{8}\\[0.1cm]
        \frac{3}{4}\\[0.1cm]
        \frac{1}{8}
    \end{pmatrix},
    \qquad
    \nae_3 =
    \begin{pmatrix}
        \frac{1}{8}\\[0.1cm]
        \frac{1}{8}\\[0.1cm]
        \frac{3}{4}
    \end{pmatrix}.
\end{align*}

Figure~\ref{fig:fp_example} illustrates the equilibrium set together with a sample trajectory of Player~1's empirical strategy under FP for the game $\tilde A$ defined in \eqref{eq:non_converge_example}. The process starts from no prior, and whenever multiple best responses arise, one of them is selected uniformly at random. Over $10^4$ iterations, the trajectory does not appear to converge to a single point in $\tildenex$. The $3$-action case already captures many of the geometric features that extend to general games $A \in \RR^{n\times m}$ and that play an essential role in our later non-convergence analysis. Throughout the paper, we use $3$-action games primarily as a convenient way to visualize these phenomena.

The key informal observation is that, once the FP trajectory enters the NE set, the dynamics retain a certain \textit{inertia} that tends to push the trajectory through the set rather than letting it settle at a single equilibrium. When the trajectory lies outside the NE set, however, the usual convergence of FP to the NE set pulls it back toward the set. To make this mechanism effective, we would like the NE set to have positive measure (Assumption~\ref{asspt:positive_measure}), so that the trajectory can spend a nontrivial amount of time inside it. Moreover, to rule out convergence to the boundary $\partial S_n$, we would like the NE set to be disjoint from $\partial S_n$ (Assumption~\ref{asspt:boundary}), that is, all equilibria are fully mixed. Under this condition, convergence to the NE set automatically rules out convergence to $\partial S_n$.

\section{Fictitious Play Cannot Converge to a Point}\label{sec:main_result}
The examples in the previous section show that there exist zero-sum games for which FP fails to converge to a point under any tie-breaking rule. A natural next step is to formulate this phenomenon in a general theorem. Our first assumption is a normalization of the game.

\setcounter{assumption}{-1}
\begin{assumption}\label{asspt:no_dominated}
For a zero-sum game with payoff matrix $A \in \RR^{n \times m}$, neither player has a strictly dominated action.
\end{assumption}

Assumption~\ref{asspt:no_dominated} is natural in view of Assumption~\ref{asspt:boundary} below: if Player~1 had a strictly dominated action, every equilibrium strategy of Player~1 would assign it zero probability, so no equilibrium strategy of Player~1 could be fully mixed. It also guarantees that every action of each player is a best response to some mixed strategy of the opponent, since in a finite game any pure action that is never a best response is strictly dominated by some mixed action~\citep{pearce1984rationalizable}.

Assumption~\ref{asspt:no_dominated} is, moreover, essentially without loss of generality. Given an arbitrary zero-sum game, remove strictly dominated actions, iteratively, until neither player has one. FP is consistent with adaptive learning in the sense of \citet{milgrom1991adaptive}; by their Lemma~4 and Theorem~5, there exists a finite time after which neither player plays any removed action.\footnote{Theorem~5 of \citet{milgrom1991adaptive} is stated for the sets $U^{\epsilon k}(S)$ surviving $k$ rounds of elimination of $\epsilon$-dominated strategies. In a finite game, the iterated elimination terminates in finitely many rounds, and each round's domination gaps are bounded away from zero by compactness; hence, for $\epsilon$ small enough and $k$ large enough, $U^{\epsilon k}(S)$ coincides with the set of actions surviving the exact iterated elimination of strictly dominated actions.} Hence the empirical strategies of the original game, projected onto the surviving actions, eventually evolve as an FP process of the reduced game, in which the finitely many plays of removed actions contribute only a fixed correction to the cumulative payoffs and can be absorbed into the initial conditions. In this sense, non-convergence of FP in the reduced game extends to the original game; Remark~\ref{rem:reduction} below makes this transfer precise.

The following two assumptions isolate the key features that were critical in the examples of Section~\ref{sec:motivation}.

\begin{assumption}\label{asspt:positive_measure}
For a zero-sum game with payoff matrix $A \in \RR^{n \times m}$, the equilibrium set of Player~1 has positive measure, i.e., $\mu(\nex)>0$.
\end{assumption}

\begin{assumption}\label{asspt:boundary}
For a zero-sum game with payoff matrix $A \in \RR^{n \times m}$, every equilibrium strategy of Player~1 is fully mixed, i.e., $\partial S_n \cap \nex=\emptyset$.
\end{assumption}

Our main result shows that these three assumptions are sufficient to rule out pointwise convergence of FP in zero-sum games.

\begin{theorem}\label{thm:non_convergence}
Suppose a zero-sum game with payoff matrix $A \in \RR^{n \times m}$ satisfies Assumptions~\ref{asspt:no_dominated},~\ref{asspt:positive_measure}, and~\ref{asspt:boundary}. If the initial empirical strategy of Player~1 lies outside the equilibrium set, i.e., $\hat\bfx(0)\notin\nex$ (in particular, if Player~1 starts with no prior information, i.e., $k_1=0$, so that $\hat\bfx(0)=\zerovector_n$), then under any tie-breaking rule, the Player~1 empirical strategy under FP cannot converge to any point in $\nex$.
\end{theorem}
Since FP converges to the NE set~\citep{Robinson} and $\nex$ is closed, every limit point of $\{\hat\bfx(k)\}$ lies in $\nex$. Combined with Theorem~\ref{thm:non_convergence}, this implies that $\{\hat\bfx(k)\}$ does not converge at all. The payoff matrices in Equations~\eqref{eq:2by2_mult_ne} and~\eqref{eq:non_converge_example} satisfy Assumptions~\ref{asspt:no_dominated}--\ref{asspt:boundary}; see Equation~\eqref{eq:non_converge_example_ne} and Figure~\ref{fig:fp_example}. Theorem~\ref{thm:non_convergence} therefore shows that the persistent oscillations observed in Section~\ref{sec:motivation} are not artifacts of the uniform tie-breaking rule used in the simulations: no tie-breaking rule can enforce pointwise convergence in these games.

\begin{remark}\label{rem:reduction}
For a general zero-sum game that does not satisfy Assumption~\ref{asspt:no_dominated}, Theorem~\ref{thm:non_convergence} applies through the above reduction. Suppose the reduced game satisfies Assumptions~\ref{asspt:positive_measure} and~\ref{asspt:boundary}, and let $T$ be a time after which neither player plays any removed action. If the empirical strategy of Player~1 at time $T$, projected onto the surviving actions, lies outside the equilibrium set of the reduced game, then the FP process from time $T$ onward satisfies the hypothesis of Theorem~\ref{thm:non_convergence}, and hence the FP trajectory of the original game cannot converge to a point.
\end{remark}

\section{Proof of Theorem~\ref{thm:non_convergence}}\label{sec:proof}
In this section, we prove Theorem~\ref{thm:non_convergence}. Hereafter, unless otherwise stated, $A \in \RR^{n \times m}$ denotes the payoff matrix of a zero-sum game. Recall that we use $v$, $\nex$, and $\ney$ to denote the value of the game, Player~1's equilibrium set, and Player~2's equilibrium set, respectively. We also assume throughout this section that the initial empirical strategy of Player~1 lies outside its equilibrium set, i.e., $\hat\bfx(0)\notin\nex$; note that this covers the no-prior case, since $k_1=0$ gives $\hat\bfx(0)=\zerovector_n\notin\nex$. Proofs omitted from this section are collected in Appendix~\ref{app:proof_oscillation}.

\subsection{Structural Properties of the Game of Interest}

Before proving non-convergence, we establish several structural properties of zero-sum games satisfying Assumptions~\ref{asspt:no_dominated}--\ref{asspt:boundary}. 

\begin{proposition}\label{prop:existence_of_one_vector}
Under Assumption~\ref{asspt:positive_measure}, there exists a column index $i \in [m]$ such that $\bfc_i = v \onevector_n$, where $v$ is the value of the game.
\end{proposition}

\begin{pf}
Since $\nex$ is convex and $\mu(\nex)>0$, its relative interior in $S_n$ is nonempty~\citep{boundary_measure_zero}. Hence there exist $\epsilon>0$ and $\bfx^*\in \nex$ such that
$B_\epsilon(\bfx^*) \cap S_n \subseteq \nex$.
Because $\bfx^* \in \nex$, there exists $i \in [m]$ such that
$\bfc_i^\top \bfx^* = v$.

Now let $\bfw \in \nullspace(\onevector_n^\top)$ satisfy $\norm{\bfw}_2=1$. Since $\bfx^* + \epsilon \bfw/2 \in \nex$, we have
\begin{align*}
\bfc_i^\top \left(\bfx^* + \frac{\epsilon}{2}\bfw\right)\ge v,
\end{align*}
because
\begin{align*}
\min_{j\in[m]} \bfc_j^\top \left(\bfx^* + \frac{\epsilon}{2}\bfw\right)=v.
\end{align*}
Using $\bfc_i^\top \bfx^*=v$, this implies $\bfc_i^\top \bfw \ge 0$. Replacing $\bfw$ by $-\bfw$, we also obtain $\bfc_i^\top \bfw \le 0$. Hence
\begin{align*}
\bfc_i^\top \bfw = 0
\qquad \text{for all } \bfw \in \nullspace(\onevector_n^\top).
\end{align*}
Therefore $\bfc_i \in \text{span}\{\onevector_n\}$, so there exists $c\in \RR$ such that
$\bfc_i = c\,\onevector_n$. Finally,
\begin{align*}
v=\bfc_i^\top \bfx^* = c\,\onevector_n^\top \bfx^* = c,
\end{align*}
since $\bfx^* \in S_n$. Thus $\bfc_i = v\onevector_n$.
\end{pf}
Henceforth, for any zero-sum game satisfying Assumption~\ref{asspt:positive_measure}, we may relabel the columns so that $\bfc_1 = v \onevector_n$.

Given that $\bfc_1=v\onevector_n$, we now show that, without loss of generality, no other column of $A$ is parallel to $\onevector_n$. Suppose $\bfc_j=v'\onevector_n$ for some $j\in[m]$ with $j\neq 1$. If $v'>v$, then column $j$ is strictly dominated for Player~2 by column~1, contradicting Assumption~\ref{asspt:no_dominated}. If $v'<v$, then the value of the game would be at most $v'$, contradicting that it is $v$. If $v'=v$, then column $j$ duplicates column~1, and removing it does not affect Player~1's FP dynamics or the associated non-convergence behavior. Thus, without loss of generality, we may assume that no column other than $\bfc_1$ is parallel to $\onevector_n$. We summarize this reduction in Assumption~\ref{asspt:deduced_from_A1}. Note that this assumption does not affect the validity of the theorem.
\begin{assumption}\label{asspt:deduced_from_A1}
Henceforth, without loss of generality, we restrict attention to zero-sum games $A \in \RR^{n \times m}$ satisfying Assumptions~\ref{asspt:no_dominated} and~\ref{asspt:positive_measure} and the following condition: the first column satisfies $\bfc_1 = v \onevector_n$, where $v$ is the value of the game, and no other column of $A$ belongs to $\text{span}\{\onevector_n\}$.
\end{assumption}

Under Assumption~\ref{asspt:deduced_from_A1}, Proposition~\ref{prop:existence_of_one_vector} implies that $\bfc_i^\top \bfx > v$ for every $i \neq 1$ and every $\bfx \in \intr(\nex)$; otherwise, applying the argument of Proposition~\ref{prop:existence_of_one_vector} at such an interior point would force $\bfc_i = v \onevector_n$, contradicting Assumption~\ref{asspt:deduced_from_A1}. Since $\bfc_1 = v \onevector_n$, it follows that $X_1$ coincides exactly with $\nex$.

To characterize $\ney$, it is useful to compare the original game with the reduced game obtained by removing the first column.
\begin{lemma}\label{lemma:value_game_minus_1}
Under Assumptions~\ref{asspt:no_dominated},~\ref{asspt:positive_measure}, and~\ref{asspt:deduced_from_A1}, let $v_{-1}$ denote the value of the reduced zero-sum game with payoff matrix $A_{-1}$. Then $v_{-1}>v$.
\end{lemma}
\begin{corollary}\label{cor:unique_NE_player2}
    Under Assumptions~\ref{asspt:no_dominated},~\ref{asspt:positive_measure}, and~\ref{asspt:deduced_from_A1}, we have $\ney=\{\bfe_1\}$, i.e., Player~2 has a unique equilibrium strategy, namely $\bfe_1$.
\end{corollary}
Next, we record a simple consequence of Assumption~\ref{asspt:boundary}.

\begin{lemma}\label{lemma:one_component_less_v}
Suppose a zero-sum game $A$ satisfies Assumptions~\ref{asspt:no_dominated}--\ref{asspt:deduced_from_A1}. Then for every $i \geq 2$, the column $\bfc_i$ has at least one entry strictly less than $v$.
\end{lemma}

We close this subsection with two properties of boundary points of $\nex$, which form the basis of the boundary analysis in Section~\ref{subsec:boundary_instability}. For $\bfw \in \RR^n$, define the spread
\begin{align*}
    \phi(\bfw) := \max_{j\in[n]} [\bfw]_j - \min_{j\in[n]} [\bfw]_j \ge 0,
\end{align*}
and note that $\phi(\bfw)=0$ if and only if $\bfw \in \text{span}\{\onevector_n\}$.

\begin{lemma}\label{lemma:binding}
Suppose $A$ satisfies Assumptions~\ref{asspt:no_dominated}--\ref{asspt:deduced_from_A1}. Then for every $\bfx^* \in \partial\nex$, the set $I:=\bry(\bfx^*)$ satisfies $1 \in I$ and $|I|\ge 2$.
\end{lemma}

\begin{lemma}\label{lemma:no_constant}
Suppose $A$ satisfies Assumptions~\ref{asspt:no_dominated}--\ref{asspt:deduced_from_A1}. Let $\bfx^* \in \partial\nex$ and $I=\bry(\bfx^*)$, and define
\begin{align*}
    W_I := \conv\bigl(\{\bfc_i : i \in I\setminus\{1\}\}\bigr).
\end{align*}
Then $W_I \cap \text{span}\{\onevector_n\} = \emptyset$. Consequently,
\begin{align*}
    \Delta_0 := \min_{\bfw\in W_I}\phi(\bfw) > 0.
\end{align*}
\end{lemma}

\begin{pf}
Suppose, for contradiction, that there exist $\theta_i\ge0$ with $\sum_{i\in I\setminus\{1\}}\theta_i=1$ and $c\in\RR$ such that
\begin{align*}
    \bfw := \sum_{i\in I\setminus\{1\}}\theta_i\,\bfc_i = c\,\onevector_n.
\end{align*}
Since $I=\bry(\bfx^*)$ and $\bfx^*\in\nex$, we have $\bfc_i^\top\bfx^*=v$ for every $i\in I$. Hence
\begin{align*}
    c = c\,\onevector_n^\top\bfx^* = \bfw^\top\bfx^* = \sum_{i\in I\setminus\{1\}}\theta_i\,\bfc_i^\top\bfx^* = v.
\end{align*}
Now let $\bar\bfy\in S_m$ be the mixed strategy with $[\bar\bfy]_i=\theta_i$ for $i\in I\setminus\{1\}$ and zero otherwise. Then $A\bar\bfy=\bfw=v\onevector_n$, so
\begin{align*}
    \max_{j\in[n]}\bfr_j^\top\bar\bfy = \max_{j\in[n]}[A\bar\bfy]_j = v,
\end{align*}
i.e., $\bar\bfy$ solves problem~\eqref{eq:lp_duality_y} and hence $\bar\bfy\in\ney$. By Corollary~\ref{cor:unique_NE_player2}, $\ney=\{\bfe_1\}$, so $\bar\bfy=\bfe_1$; but $[\bar\bfy]_1=0$ by construction, a contradiction. The final claim follows since $W_I$ is a nonempty (Lemma~\ref{lemma:binding}) compact set on which the continuous function $\phi$ is strictly positive.
\end{pf}

\subsection{Non-Convergence to Interior Points of \texorpdfstring{$\nex$}{NEx}}
We now show that FP cannot converge to any point in the interior of $\nex$. Before turning to the proof, we outline the main intuition and establish several supporting properties.
The basic intuition is that, once the trajectory remains close to $\nex$, Player~1 must play every action a linear number of times, while Corollary~\ref{cor:unique_NE_player2} implies that Player~2 plays Action~1 for a linear number of times. This forces the dynamics to revisit the regions $X_1 \times Y_i$ repeatedly. The following lemma makes this observation precise.

\begin{lemma}\label{lemma:infinite_reach}
    Under Assumptions~\ref{asspt:no_dominated}--\ref{asspt:deduced_from_A1}, each region $X_1 \times Y_1,\dots,X_1 \times Y_n$ is visited $\Omega(T)$ times.
\end{lemma}

\begin{pf}
Consider the FP dynamics $\{(\hat\bfx(k),\hat\bfy(k))\}_{k=1}^T$. Because $\nex$ is compact and $\partial S_n \cap \nex=\emptyset$, there exists $\delta>0$ such that every $\bfx\in \nex$ satisfies $\bfx>\delta\onevector_n$ coordinatewise. Since FP converges to the NE set, for all sufficiently large $T$, Player~1 must play each action at least $\delta T/2$ times up to time $T$.

On the other hand, since $\ney=\{\bfe_1\}$ (Corollary~\ref{cor:unique_NE_player2}) and FP converges to the NE set, Player~2 plays Action~1 for $T - o(T)$ times and all other actions for only $o(T)$ times. Hence, for each $i\in[n]$, there are $\Omega(T)$ time steps at which Player~1 plays Action~$i$ and Player~2 simultaneously plays Action~1.

Recall that Player~1 plays Action~$i$ only when $\hat\bfy(k)\in Y_i$, and Player~2 plays Action~1 only when $\hat\bfx(k)\in X_1$. Therefore, each region $X_1\times Y_i$ is visited $\Omega(T)$ times.
\end{pf}

The next lemma formalizes the inertia phenomenon: as long as $\hat\bfx(k)$ remains in $\intr(X_1)$, Player~2 keeps playing Action~1, and the best-response set of Player~1 does not change.

\begin{lemma}\label{lemma:inertia}
    Under Assumptions~\ref{asspt:no_dominated}--\ref{asspt:deduced_from_A1}, if $(\hat\bfx(k),\hat\bfy(k)) \in \intr(X_1)\times Y_i^c$ for some $i\in[n]$, then $\hat\bfy(k+1)\in Y_i^c$.
\end{lemma}

\begin{pf}
    Since $\hat\bfx(k)\in \intr(X_1)$, Player~2 plays Action~1 at time $k$. Hence
    \begin{align}
        \brx(\hat\bfy(k+1))
        &= \argmax_{j\in[n]} [A\hat\bfy(k+1)]_j \nonumber\\
        &= \argmax_{j\in[n]} [(k+k_2+1)A\hat\bfy(k+1)]_j \nonumber\\
        &= \argmax_{j\in[n]} [(k+k_2)A\hat\bfy(k)+\bfc_1]_j \nonumber\\
        &= \argmax_{j\in[n]} [(k+k_2)A\hat\bfy(k)]_j \nonumber\\
        &= \brx(\hat\bfy(k)),
        \label{eq:inertia}
    \end{align}
    where the third line uses
    \begin{align*}
        (k+k_2+1)\hat\bfy(k+1)=(k+k_2)\hat\bfy(k)+\bfe_1,
    \end{align*}
    and the fourth line follows from $\bfc_1=v\onevector_n$.

    Since $\hat\bfy(k)\in Y_i^c$, we have $i\notin \brx(\hat\bfy(k))$. By \eqref{eq:inertia}, it follows that $i\notin \brx(\hat\bfy(k+1))$, and hence $\hat\bfy(k+1)\in Y_i^c$.
\end{pf}

As a consequence, if $\hat\bfx(k)\in \intr(X_1)$ for all $k\in [K_1,K_2)$ with $K_1<K_2$, then the best-response set of Player~1 remains fixed throughout this interval and is equal to $\brx(\hat\bfy(K_1))$.

\begin{lemma}\label{lemma:conditional_instability}
Under Assumptions~\ref{asspt:no_dominated}--\ref{asspt:deduced_from_A1}, if $(\hat\bfx(k),\hat\bfy(k)) \in \intr(X_1)\times Y_i^c$ for some $i\in[n]$, then there exists $k'>k$ such that $(\hat\bfx(k'),\hat\bfy(k')) \in \intr(X_1)^c\times Y_i^c$. In particular, Player~2 must eventually choose an action other than Action~1.
\end{lemma}

\begin{pf}
By Lemma~\ref{lemma:infinite_reach}, the region $X_1\times Y_i$ is visited infinitely often. Suppose, for contradiction, that Player~2 continues to play Action~1 at every time $t\ge k$. Then the inertia Eq.~\eqref{eq:inertia} implies inductively that
\begin{align*}
    \hat\bfy(t)\in Y_i^c
    \qquad \text{for all } t\ge k.
\end{align*}
Hence $(\hat\bfx(t),\hat\bfy(t))\notin X_1\times Y_i$ for all $t\ge k$, contradicting Lemma~\ref{lemma:infinite_reach}. Therefore, Player~2 must choose an action other than Action~1 at some time after $k$.

Let $k'>k$ be the first time at which Player~2 chooses an action other than Action~1. Since Player~2 plays Action~1 at all times $t\in\{k,\dots,k'-1\}$, repeated application of Lemma~\ref{lemma:inertia} yields
\begin{align*}
    \hat\bfy(k')\in Y_i^c.
\end{align*}
Moreover, because Player~2 does not choose Action~1 at time $k'$, we must have $\hat\bfx(k')\notin \intr(X_1)$; otherwise, if $\hat\bfx(k')\in \intr(X_1)$, then Player~2's unique best response would be Action~1. Thus
\begin{align*}
    (\hat\bfx(k'),\hat\bfy(k'))\in \intr(X_1)^c\times Y_i^c.
\end{align*}
\end{pf}

Lemma~\ref{lemma:conditional_instability} shows that if $\hat\bfy(k)\notin \bigcap_{i\in[n]} Y_i$, then the FP trajectory cannot remain in $\intr(\nex)$ forever. Equivalently, whenever $\hat\bfx(k)\in \intr(\nex)$ and $\hat\bfy(k)\notin \bigcap_{i\in[n]} Y_i$, there exists a later time $k'\ge k$ at which $\hat\bfx(k')\in \intr(\nex)^c$. The following proposition provides a condition ensuring that $\hat\bfy(k)\notin \bigcap_{i\in[n]} Y_i$ throughout the FP trajectory.

\begin{proposition}\label{prop:interior_instability}
Under Assumptions~\ref{asspt:no_dominated}--\ref{asspt:deduced_from_A1}, if there exists $k'\ge 0$ such that $\hat\bfx(k') \in \{\zerovector_n\}\cup (S_n\setminus X_1)$, then the region $X_1 \times \bigcap_{i \in [n]} Y_i$ cannot be visited at any later time $k>k'$.
\end{proposition}

\begin{pf}
Suppose, for contradiction, that $k>k'$ is the earliest time such that
\begin{align*}
    \hat\bfx(k)\in X_1
    \qquad\text{and}\qquad
    \hat\bfy(k)\in \bigcap_{i\in[n]} Y_i.
\end{align*}
We first observe that $k-1+k_1\ge1$, so that $\hat\bfx(k-1)\in S_n$. Indeed, this is immediate if $k_1>0$ or $k'\ge1$. In the remaining case $k_1=0$ and $k'=0$, we must have $k\ge2$: if $k=1$, then $\hat\bfx(1)\in\partial S_n$, while Assumption~\ref{asspt:boundary} and $X_1=\nex$ give $\partial S_n\cap X_1=\emptyset$, contradicting $\hat\bfx(1)\in X_1$.

Since $\hat\bfy(k)\in \bigcap_{i\in[n]} Y_i$, we have
\begin{align*}
    \bfr_1^\top \hat\bfy(k)=\cdots=\bfr_n^\top \hat\bfy(k).
\end{align*}

\noindent\textbf{Case 1:} Suppose Player~2 plays Action~$i\ge 2$ at time $k-1$. Then
\begin{align*}
    (k-1+k_2)\hat\bfy(k-1)+\bfe_i=(k+k_2)\hat\bfy(k),
\end{align*}
and hence
\begin{align*}
    \brx(\hat\bfy(k-1))
    &= \argmax_{j\in[n]} \bfr_j^\top \hat\bfy(k-1) \\
    &= \argmax_{j\in[n]} \bfr_j^\top\bigl((k+k_2)\hat\bfy(k)-\bfe_i\bigr) \\
    &= \argmin_{j\in[n]} a_{ji}.
\end{align*}
Let $j'$ be the action chosen by Player~1 at time $k-1$. Then
\begin{align*}
    j'\in \argmin_{j\in[n]} a_{ji}.
\end{align*}
By Lemma~\ref{lemma:one_component_less_v}, we have $a_{j'i}<v$.
Since Player~2 chose Action~$i$ at time $k-1$, we also have
\begin{align*}
    \bfc_i^\top \hat\bfx(k-1)
    = \min_{l\in[m]} \bfc_l^\top \hat\bfx(k-1)
    \le \bfc_1^\top \hat\bfx(k-1)
    = v.
\end{align*}
Because $(k+k_1)\hat\bfx(k)=(k-1+k_1)\hat\bfx(k-1)+\bfe_{j'}$,
\begin{align*}
    \bfc_i^\top \hat\bfx(k)
    = \frac{(k-1+k_1)\,\bfc_i^\top \hat\bfx(k-1)+a_{j'i}}{k+k_1}
    < v.
\end{align*}
This contradicts $\hat\bfx(k)\in X_1$, since $\bfc_1^\top \hat\bfx(k)=v$.

\noindent\textbf{Case 2:} Suppose Player~2 plays Action~1 at time $k-1$. Then $\hat\bfx(k-1)\in X_1$. Moreover, by the same calculation as in \eqref{eq:inertia},
\begin{align*}
    \brx(\hat\bfy(k-1))=\brx(\hat\bfy(k)).
\end{align*}
Since $\hat\bfy(k)\in \bigcap_{i\in[n]} Y_i$, we have $\brx(\hat\bfy(k))=[n]$, and hence $\brx(\hat\bfy(k-1))=[n]$ as well. Therefore
\begin{align*}
    \hat\bfy(k-1)\in \bigcap_{i\in[n]} Y_i.
\end{align*}
It follows that
\begin{align*}
    (\hat\bfx(k-1),\hat\bfy(k-1))\in X_1\times \bigcap_{i\in[n]} Y_i,
\end{align*}
contradicting the minimality of $k$ if $k-1>k'$, and contradicting the hypothesis $\hat\bfx(k')\in \{\zerovector_n\}\cup(S_n\setminus X_1)$ if $k-1=k'$.
\end{pf}

Recall that $\hat\bfx(0)\notin\nex=X_1$: if $k_1>0$, then $\hat\bfx(0)\in S_n\setminus X_1$, and if $k_1=0$, then $\hat\bfx(0)=\zerovector_n$. Hence the hypothesis of Proposition~\ref{prop:interior_instability} is satisfied with $k'=0$, and the region $X_1 \times \bigcap_{i\in[n]} Y_i$ can never be visited after time $0$. Consequently, whenever the trajectory enters $\intr(\nex)=\intr(X_1)$, the component $\hat\bfy(k)$ must lie outside $\bigcap_{i\in[n]} Y_i$, so Lemma~\ref{lemma:conditional_instability} implies that the trajectory cannot remain in $\intr(\nex)$ forever. This rules out convergence to any point in $\intr(\nex)$. Moreover, since Lemma~\ref{lemma:infinite_reach} shows that the regions $X_1\times Y_i$ are visited infinitely often, repeated application of Lemma~\ref{lemma:conditional_instability} implies that Player~2 must choose an action other than Action~1 infinitely often. We call such a time step, at which Player~2 plays an action other than its unique equilibrium action (Action~1), a \textit{deviation}.

\subsection{Instability at Boundary Points of \texorpdfstring{$\nex$}{NEx}}\label{subsec:boundary_instability}
We now rule out convergence to boundary points of $\nex$. In contrast with the interior case, the argument is a counting argument. The key observation is that there are two ways to bound the length of a time window within which Player~1 must play every action. The first comes from pointwise convergence itself: if $\hat\bfx(k)$ converges to a fully mixed point, then every action must be played within a window whose length is proportional to the elapsed time (Lemma~\ref{lemma:quota}). The second comes from the structure of cumulative payoffs. An action is played only when its cumulative payoff is maximal, and since Player~2's deviations from Action~1 never reward all of Player~1's actions equally (Lemma~\ref{lemma:no_constant}), the deviations drive the cumulative payoffs apart: the gap of the most disadvantaged action grows in proportion to the total number of deviations, while each deviation can close the gap only by a bounded amount. Hence a lagging action must \textit{climb} back to the maximum in bounded increments, and the window required for this climb grows strictly faster than the window permitted by pointwise convergence. The proof below makes this tension precise.

\begin{lemma}\label{lemma:quota}
Suppose $\hat\bfx(k)\to\bfx^*$ with $\delta^*:=\min_{j\in[n]}[\bfx^*]_j>0$. Then for every $\beta\in(0,1]$ there exists $T_*\in\NN$ such that, for every $T\ge T_*$ and every $j\in[n]$, Player~1 plays Action~$j$ at least once during the time interval $[T,\lceil(1+\beta)T\rceil)$.
\end{lemma}

\begin{pf}
By the update rule~\eqref{eq:dynamics}, the number of times Action~$j$ is played during $[T,T')$, where $T'=\lceil(1+\beta)T\rceil$, equals $(T'+k_1)[\hat\bfx(T')]_j-(T+k_1)[\hat\bfx(T)]_j$. Let $\bar e(T):=\sup_{t\ge T}\norm{\hat\bfx(t)-\bfx^*}_\infty$, which tends to zero as $T\to\infty$. Then, for $T\ge\max(k_1,1)$,
\begin{align*}
    (T'+k_1)[\hat\bfx(T')]_j - (T+k_1)[\hat\bfx(T)]_j
    &\ge (T'-T)[\bfx^*]_j - 2(T'+k_1)\,\bar e(T)\\
    &\ge \beta T\delta^* - 8(1+\beta)T\,\bar e(T),
\end{align*}
which is strictly positive once $\bar e(T) < \beta\delta^*/\bigl(16(1+\beta)\bigr)$.
\end{pf}

\begin{proposition}\label{prop:boundary_instability}
Under Assumptions~\ref{asspt:no_dominated}--\ref{asspt:deduced_from_A1}, if $\hat\bfx(0)\notin\nex$, then the Player~1 empirical strategy under FP cannot converge to any point $\bfx^*\in\partial\nex$.
\end{proposition}

\begin{pf}
Suppose, for contradiction, that $\hat\bfx(k)\to\bfx^*\in\partial\nex$, and let $I=\bry(\bfx^*)$, so that $1\in I$ and $|I|\ge2$ by Lemma~\ref{lemma:binding}.

\noindent\textbf{Step 1 (localization).}
For every $j\notin I$ we have $\bfc_j^\top\bfx^*>v$, so by continuity there exist $\ep>0$ and $K_0\in\NN$ such that $\bfc_j^\top\bfx>v$ for all $j\notin I$ and all $\bfx\in B_\ep(\bfx^*)\cap S_n$, and $\hat\bfx(k)\in B_\ep(\bfx^*)\cap S_n$ for all $k\ge K_0$. Hence $\bry(\hat\bfx(k))\subseteq I$ for all $k\ge K_0$, i.e., Player~2 plays only actions in $I$ from time $K_0$ on. Let $a_t$ denote the action chosen by Player~2 at time $t$. In particular, every deviation at time $t\ge K_0$ satisfies $a_t\in I\setminus\{1\}$. Write, for $T\ge K_0$,
\begin{align*}
    \mathcal{D}(T) := \left|\bigl\{t \in [K_0, T) : a_t \in I\setminus\{1\}\bigr\}\right|.
\end{align*}
By the paragraph following Proposition~\ref{prop:interior_instability} (which uses $\hat\bfx(0)\notin\nex$), Player~2 deviates from Action~1 infinitely often, so $\mathcal{D}(T)\to\infty$.

\noindent\textbf{Step 2 (evolution of cumulative payoffs).}
By the update rule~\eqref{eq:dynamics},
\begin{align*}
    (T+k_2)\hat\bfy(T) = (K_0+k_2)\hat\bfy(K_0) + \sum_{t=K_0}^{T-1}\bfe_{a_t},
\end{align*}
so, for $T\ge K_0$, Player~1's cumulative payoff vector decomposes as
\begin{align*}
    (T+k_2)A\hat\bfy(T) = \mathbf{u}(K_0) + N_1(T)\,v\onevector_n + \mathbf{s}(T),
\end{align*}
where $\mathbf{u}(K_0):=(K_0+k_2)A\hat\bfy(K_0)$, $N_1(T)$ is the number of plays of Action~1 during $[K_0,T)$ (recall $\bfc_1=v\onevector_n$), and
\begin{align*}
    \mathbf{s}(T) := \sum_{\substack{K_0\le t<T\\ a_t\in I\setminus\{1\}}}\bfc_{a_t}.
\end{align*}
Since adding multiples of $\onevector_n$ does not change the set of maximizers,
\begin{align}\label{eq:score_argmax}
    \brx(\hat\bfy(T)) = \argmax_{j\in[n]}\,[\mathbf{q}(T)]_j,
    \qquad
    \mathbf{q}(T):=\mathbf{u}(K_0)+\mathbf{s}(T).
\end{align}
In other words, $\mathbf{q}(T)$ is effectively Player~1's cumulative payoff vector: it differs from $(T+k_2)A\hat\bfy(T)$ only by a multiple of $\onevector_n$, and hence determines Player~1's best responses.

\noindent\textbf{Step 3 (spread lower bound).}
Whenever $\mathcal{D}(T)\ge1$, define the average deviation direction
\begin{align*}
    \bar\bfw(T)
    :=
    \frac{1}{\mathcal{D}(T)}
    \sum_{\substack{K_0\le t<T\\ a_t\in I\setminus\{1\}}}
    \bfc_{a_t}
    = \frac{\mathbf{s}(T)}{\mathcal{D}(T)}
    \;\in\;
    W_I.
\end{align*}
Since $\phi$ is positively homogeneous, Lemma~\ref{lemma:no_constant} gives
\begin{align*}
    \phi(\mathbf{s}(T)) = \mathcal{D}(T)\,\phi\bigl(\bar\bfw(T)\bigr) \ge \Delta_0\,\mathcal{D}(T).
\end{align*}
Setting $C_0:=\norm{\mathbf{u}(K_0)}_\infty$ and using $\phi(\bfx+\bfy)\ge\phi(\bfx)-2\norm{\bfy}_\infty$,
\begin{align}\label{eq:score_spread}
    \max_{j\in[n]} g_j(T) = \phi(\mathbf{q}(T)) \ge \Delta_0 \mathcal{D}(T)-2C_0,
\end{align}
where $g_j(T):=\max_{l\in[n]}[\mathbf{q}(T)]_l-[\mathbf{q}(T)]_j\ge0$ is the cumulative payoff gap between the best actions and Action~$j$. If $g_j(T)=0$, then Action~$j$ attains the maximum cumulative payoff, i.e., it is among the best responses of Player~1; conversely, by \eqref{eq:score_argmax}, Player~1 can play Action~$j$ at time $T$ only if $g_j(T)=0$.

\noindent\textbf{Step 4 (climb cost).}
The vector $\mathbf{q}$ changes only at deviations, and a deviation to column $i$ changes each gap $g_j$ by at most $\phi(\bfc_i)$ in absolute value, since
\begin{align*}
    g_j(t+1)-g_j(t) = \Bigl(\max_l\bigl([\mathbf{q}(t)]_l+[\bfc_i]_l\bigr)-\max_l[\mathbf{q}(t)]_l\Bigr)-[\bfc_i]_j
    \in\bigl[-\phi(\bfc_i),\,\phi(\bfc_i)\bigr].
\end{align*}
Let $\kappa:=\max_{i\in I\setminus\{1\}}\phi(\bfc_i)\ge\Delta_0>0$. Then, for any $T\le\tau$ and any $j$ with $g_j(T)\ge G$ and $g_j(\tau)=0$, the number of deviations during $[T,\tau)$ is at least $G/\kappa$.

\noindent\textbf{Step 5 (slow climbing).}
Set $d:=\Delta_0/(2\kappa)\in(0,1/2]$ and $\beta:=d/4$. Recall that, since $\bfx^*\in\nex$ and Assumption~\ref{asspt:boundary} gives $\nex\cap\partial S_n=\emptyset$, we have $\delta^*=\min_j[\bfx^*]_j>0$, so Lemma~\ref{lemma:quota} applies with this $\beta$ and yields a corresponding time $T_*$. Choose $T_0\ge\max(T_*,\,4/d,\,K_0)$ large enough that $\mathcal{D}(T_0)\ge\max(1,\,4C_0/\Delta_0)$, and define $T_{N+1}:=\lceil(1+\beta)T_N\rceil$ for $N\ge0$, inductively; note $T_{N+1}\le(1+d/2)T_N$ because $T_N\ge T_0\ge 4/d$.

Fix $N\ge0$, and let $j^* \in \argmin_j[\mathbf{q}(T_N)]_j$. By \eqref{eq:score_spread} and $\mathcal{D}(T_N)\ge \mathcal{D}(T_0)\ge4C_0/\Delta_0$,
\begin{align*}
    g_{j^*}(T_N) = \phi(\mathbf{q}(T_N)) \ge \Delta_0\mathcal{D}(T_N)-2C_0 \ge \tfrac{\Delta_0}{2}\,\mathcal{D}(T_N).
\end{align*}
By Lemma~\ref{lemma:quota}, Player~1 plays Action~$j^*$ at some time $\tau\in[T_N,T_{N+1})$, so $g_{j^*}(\tau)=0$, and Step~4 yields at least $\tfrac{\Delta_0}{2\kappa}\mathcal{D}(T_N)=d\,\mathcal{D}(T_N)$ deviations during $[T_N,\tau)\subseteq[T_N,T_{N+1})$. Hence
\begin{align*}
    \mathcal{D}(T_{N+1}) \ge (1+d)\,\mathcal{D}(T_N)
    \qquad\text{for all }N\ge0,
\end{align*}
and therefore $\mathcal{D}(T_N)\ge(1+d)^N$. On the other hand, $\mathcal{D}(T_N)\le T_N\le(1+d/2)^NT_0$, so
\begin{align*}
    \Bigl(\frac{1+d}{1+d/2}\Bigr)^N \le T_0
    \qquad\text{for all }N\ge0,
\end{align*}
which is impossible since the left-hand side tends to infinity. This contradiction completes the proof.
\end{pf}

Every point of $\nex$ lies either in $\intr(\nex)$ or in $\partial\nex$. Convergence to a point of $\intr(\nex)$ was ruled out at the end of the previous subsection, and Proposition~\ref{prop:boundary_instability} rules out convergence to a point of $\partial\nex$. Since the reduction to Assumption~\ref{asspt:deduced_from_A1} does not affect the conclusion, this completes the proof of Theorem~\ref{thm:non_convergence}.

We close with the conceptual message of the proof: the fact that Player~2's unique NE (Corollary~\ref{cor:unique_NE_player2}) is a pure action is precisely what prevents equilibrium selection on Player~1's side. A fully mixed limit point would require Player~2's deviations to keep all of Player~1's actions indifferent on average, so that lagging actions could climb back at bounded cost. Lemma~\ref{lemma:no_constant} rules this out: any such indifference-preserving mixture would constitute a second equilibrium strategy of Player~2.

\section{Discussion}\label{sec:discussion}
\subsection{Necessity of the Assumptions}\label{subsec:necessity}
We emphasize that violating Assumptions~\ref{asspt:positive_measure} and~\ref{asspt:boundary} does not by itself imply pointwise non-convergence: there are zero-sum games that violate one or both of these assumptions and whose FP dynamics still converge, or are conjectured to converge, pointwise. For instance, the game with payoff matrix $\tilde A$ in Equation~\eqref{eq:zz_mat}, introduced in Section~\ref{sec:tie_non_conv}, fails to satisfy Assumption~\ref{asspt:boundary} and yet exhibits pointwise convergence under several natural tie-breaking rules.

\begin{figure}
  \centering
  \includegraphics[width=\linewidth]{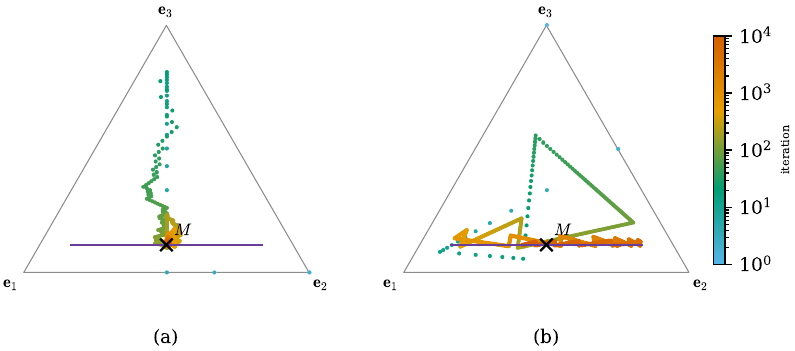}
  \caption{Two sample FP trajectories of Player~1 for the game in \eqref{eq:converging_example}, which violates Assumption~\ref{asspt:positive_measure}: its equilibrium set is the measure-zero purple segment with midpoint $M$ (black cross). Both runs start from no prior and break ties uniformly at random; each run is $10^4$ iterations, and the color indicates the iteration. (a)~A realization that appears to settle near $M$. (b)~A realization that keeps drifting along the equilibrium segment.}
  \label{fig:lower_dim_convergence}
\end{figure}

\begin{figure}
  \centering
  \includegraphics[width=\linewidth]{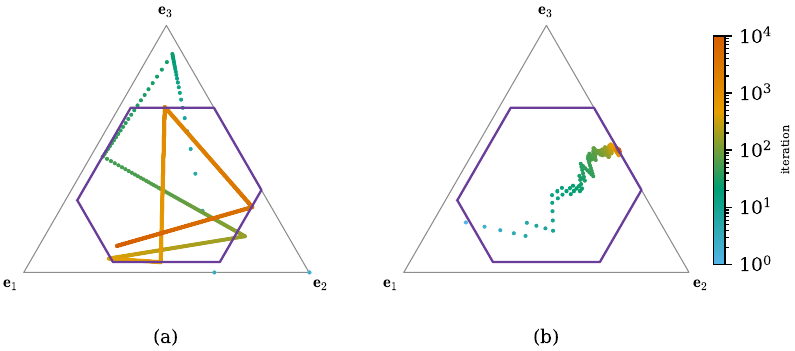}
  \caption{Sample FP trajectories of Player~1 for the game in \eqref{eq:without_a2}, which violates Assumption~\ref{asspt:boundary}: its equilibrium set (purple hexagon) intersects $\partial S_3$. Ties are broken uniformly at random; each run is $10^4$ iterations, and the color indicates the iteration. (a)~Initialized with no prior, the trajectory does not appear to converge. (b)~Initialized as in Equation~\eqref{eq:without_a2_conv}, the trajectory appears to converge.}
  \label{fig:without_a2}
\end{figure}


The behavior of FP for games that do not satisfy Assumption~\ref{asspt:positive_measure} is more delicate. Figure~\ref{fig:lower_dim_convergence} illustrates two numerical trajectories for the zero-sum game
\begin{align}\label{eq:converging_example}
\tilde A = \begin{pmatrix} 
    1/8 & 1 & 0 & 0 \\ 
    1/8 & 0 & 1 & 0 \\ 
    0 & 0 & 0 & 1 
    \end{pmatrix}.
\end{align}
This game has infinitely many equilibria and satisfies Assumption~\ref{asspt:boundary}, but not Assumption~\ref{asspt:positive_measure}. When ties are broken uniformly at random among multiple best responses, the trajectory is quickly attracted to the equilibrium segment. The subsequent behavior, however, differs across realizations: in some runs the empirical strategy appears to settle near the midpoint $M$ of the segment (Figure~\ref{fig:lower_dim_convergence}(a)), while in others it continues to drift along the segment (Figure~\ref{fig:lower_dim_convergence}(b)). Whether pointwise convergence occurs for this game remains unclear.

We also suspect that Assumptions~\ref{asspt:positive_measure} and~\ref{asspt:boundary} are not strictly necessary for non-convergence under particular initializations. Consider the zero-sum game
\begin{align}\label{eq:without_a2}
\tilde A = \begin{pmatrix} 
    1/24 & 0 & 0 & 1/8 & 0 & 1/8 \\ 
    1/24 & 1 & 0 & 1/8 & 1/8 & 0 \\ 
    1/24 & 0 & 1 & 0 & 1/8 & 1/8 
    \end{pmatrix}.
\end{align}
For this game, the equilibrium set $\tildenex$ intersects $\partial S_3$, so Assumption~\ref{asspt:boundary} fails. Nevertheless, numerical experiments suggest that FP need not converge for some initializations. For example, when Player~1 starts with no prior, i.e., $k_1=0$, the trajectory appears not to converge; see Figure~\ref{fig:without_a2}(a). On the other hand, the same game also admits initial conditions for which FP appears to converge. One such example is
\begin{align}\label{eq:without_a2_conv}
\hat \bfx(0) = \begin{pmatrix}
    \frac 3 4\\[0.05cm]\frac 1 8\\[0.05cm] \frac 1 8
\end{pmatrix}, \quad \hat \bfy(0) = \frac 1 6 \onevector_6, \quad k_1 = k_2 = 10,
\end{align}
for which the trajectory appears to converge; see Figure~\ref{fig:without_a2}(b).

\subsection{What Can Serve as a Stopping Criterion for Learning?}\label{subsec:stopping}

Part of the appeal of FP is that learning requires no knowledge of the payoff matrix: the dynamics are driven entirely by cumulative payoffs, so each player needs only the ability to choose, through experience, the action that would have performed best against the opponent's past play. This feedback is purely ordinal: the player ranks its actions by past performance, but never observes cardinal payoffs. In zero-sum games, this minimal feedback already suffices to learn equilibrium, in the sense that the empirical strategies converge to the NE set~\citep{Robinson}. This guarantee, however, is available only to an outside observer: a player who does not know the payoff matrix cannot verify that the current empirical strategies form a NE. A natural question therefore arises: by what observable criterion can a player recognize that equilibrium has been learned?

Perhaps the most tempting candidate is stabilization---declare learning complete once one's own empirical strategy, which each player can observe, stops moving. Our results show that this criterion is unreliable in a strong sense: in any game satisfying Assumptions~\ref{asspt:no_dominated}--\ref{asspt:boundary}, the empirical strategy of Player~1, initialized outside the equilibrium set, never stabilizes under any tie-breaking rule, even though it converges to the equilibrium set. A player waiting for its empirical strategy to settle would thus wait forever, despite having already learned the equilibrium in the set-convergence sense; moreover, the strategy it holds at any stopping time remains an artifact of that time. Therefore, at a minimum, pointwise convergence cannot serve as a stopping criterion, and identifying observable criteria that can is an interesting direction for future research.

\section{Conclusion}
In this paper, we studied pointwise convergence of fictitious play in zero-sum games. While classical results show that fictitious play converges to the Nash equilibrium set, our results show that this set convergence need not imply convergence to a single equilibrium point. In particular, we proved that whenever, in a game without strictly dominated actions, the equilibrium set of a player has positive measure and consists only of fully mixed strategies (Assumptions~\ref{asspt:no_dominated}--\ref{asspt:boundary}), fictitious play cannot converge pointwise for that player under any tie-breaking rule, provided the player's initial empirical strategy lies outside the equilibrium set, as is the case when the player starts with no prior information. The proof combines an inertia mechanism in the interior of the equilibrium set with a counting argument at its boundary, which rests on the observation that the opponent's deviations from its unique equilibrium action never reward all of the player's actions equally.

Several questions remain open. Our assumptions are sufficient but, as Section~\ref{subsec:necessity} illustrates, may not be necessary, and an exact characterization of the games and initializations for which pointwise convergence fails remains to be found. Moreover, since pointwise convergence cannot serve as a stopping criterion, it remains unclear by what observable criterion a player can recognize that equilibrium has been learned (Section~\ref{subsec:stopping}). It would also be interesting to investigate analogous selection questions for other classes of games and learning dynamics. We hope that the viewpoint developed here helps clarify the gap between convergence to the equilibrium set and convergence to a point in fictitious play.

\section*{Declaration of generative AI and AI-assisted technologies in the writing process}
During the preparation of this work, the author used Claude (Anthropic) to
assist with drafting and revising the manuscript, and as an aid in
exploring and formalizing parts of the boundary-instability argument of
Section~\ref{subsec:boundary_instability}. The author verified all
arguments, reviewed and edited all content, and takes full responsibility
for the content of the publication.

\appendix

\section{Omitted Proofs}\label{app:proof_oscillation}

\begin{pf}[of Proposition~\ref{prop:oscillation}]
The claim is immediate for $n=1$. We first show that \eqref{eq:odd_case} implies \eqref{eq:even_case}. Suppose that
\begin{align*}
    [\hat\bfx(2^n)]_2=\frac{1}{2}
\end{align*}
for some $n\ge 1$. Since the numerator of $1/2$ in lowest terms is odd, the tie-breaking rule selects Action~2 at time $2^n$. We claim that
\begin{align*}
    [\hat\bfx(2^n+k)]_2=\frac{2^{n-1}+k}{2^n+k}
    \qquad \text{for } 0\le k\le 2^{n-1}.
\end{align*}
The case $k=0$ is exactly \eqref{eq:odd_case}. Now suppose the claim holds for some $k<2^{n-1}$, and write
\begin{align*}
    \frac{2^{n-1}+k}{2^n+k}=\frac{p}{q}
\end{align*}
in lowest terms. Since
\begin{align*}
    (2^n+k)-(2^{n-1}+k)=2^{n-1},
\end{align*}
every power of $2$ dividing $2^{n-1}+k$ also divides $2^n+k$. Hence, after reduction, the numerator $p$ is odd, so the tie-breaking rule again selects Action~2. Therefore,
\begin{align*}
    [\hat\bfx(2^n+k+1)]_2
    =\frac{2^{n-1}+k+1}{2^n+k+1}.
\end{align*}
This proves the claim. Setting $k=2^{n-1}$ gives \eqref{eq:even_case}.

Next, we show that \eqref{eq:even_case} implies \eqref{eq:odd_case} with $n$ replaced by $n+1$. Suppose that
\begin{align*}
    [\hat\bfx(2^n+2^{n-1})]_2=\frac{2}{3}.
\end{align*}
We claim that
\begin{align*}
    [\hat\bfx(2^n+2^{n-1}+k)]_2
    =\frac{2^n}{2^n+2^{n-1}+k}
    \qquad \text{for } 0\le k\le 2^{n-1}.
\end{align*}
Again, the case $k=0$ is exactly \eqref{eq:even_case}. Suppose the claim holds for some $k<2^{n-1}$, and write
\begin{align*}
    \frac{2^n}{2^n+2^{n-1}+k}=\frac{p}{q}
\end{align*}
in lowest terms. Since
\begin{align*}
    2^n < 2^n+2^{n-1}+k < 2^{n+1},
\end{align*}
the denominator is not divisible by $2^n$. Hence the reduced numerator $p$ remains even, and the tie-breaking rule selects Action~1. Therefore,
\begin{align*}
    [\hat\bfx(2^n+2^{n-1}+k+1)]_2
    =\frac{2^n}{2^n+2^{n-1}+k+1}.
\end{align*}
This proves the claim. Setting $k=2^{n-1}$ yields
\begin{align*}
    [\hat\bfx(2^{n+1})]_2=\frac{1}{2},
\end{align*}
which is \eqref{eq:odd_case} with $n$ replaced by $n+1$.
\end{pf}

\begin{pf}[of Lemma~\ref{lemma:value_game_minus_1}]
Suppose, for contradiction, that $v_{-1}\le v$. Let $\nex^{(-1)}$ denote Player~1's equilibrium set for the reduced game $A_{-1}$. If $\mu(\nex^{(-1)})>0$, then Proposition~\ref{prop:existence_of_one_vector} applied to $A_{-1}$ would imply that some column of $A_{-1}$ is equal to $v_{-1}\onevector_n$, contradicting Assumption~\ref{asspt:deduced_from_A1}. Therefore $\mu(\nex^{(-1)})=0$.

By definition of the value $v_{-1}$, for every $\bfx\in S_n$,
\begin{align*}
    \min_{i\in\{2,\dots,m\}} \bfc_i^\top \bfx \le v_{-1},
\end{align*}
and equality holds if and only if $\bfx\in \nex^{(-1)}$. Hence, for almost every $\bfx\in S_n$,
\begin{align*}
    \min_{i\in\{2,\dots,m\}} \bfc_i^\top \bfx < v_{-1}\le v.
\end{align*}
Since $\bfc_1=v\onevector_n$, we have $\bfc_1^\top \bfx=v$ for every $\bfx\in S_n$. Therefore, for almost every $\bfx\in S_n$,
\begin{align*}
    \min_{i\in[m]} \bfc_i^\top \bfx < v.
\end{align*}
But $\bfx\in \nex$ if and only if $\min_{i\in[m]} \bfc_i^\top \bfx=v$, contradicting Assumption~\ref{asspt:positive_measure}, which states that $\mu(\nex)>0$. Thus $v_{-1}>v$.
\end{pf}

\begin{pf}[of Corollary~\ref{cor:unique_NE_player2}]
    Let $\bfy^* \in \ney$, and write $\alpha = [\bfy^*]_1$. Since $\bfy^*$ is an equilibrium strategy of Player~2, we have
    \begin{align*}
        v = \max_{j \in [n]} \bfr_j^\top \bfy^* 
          = v\alpha + \max_{j \in [n]} (\bfr_j)_{-1}^\top (\bfy^*)_{-1}.
    \end{align*}
    If $\alpha < 1$, define
    \begin{align*}
        \bar \bfy := \frac{(\bfy^*)_{-1}}{1-\alpha} \in S_{m-1}.
    \end{align*}
    Then
    \begin{align*}
        \max_{j \in [n]} (\bfr_j)_{-1}^\top (\bfy^*)_{-1}
        = (1-\alpha)\max_{j \in [n]} (\bfr_j)_{-1}^\top \bar \bfy
        \ge (1-\alpha) v_{-1},
    \end{align*}
    where the last inequality follows from the definition of $v_{-1}$. Hence
    \begin{align*}
        v \ge v\alpha + (1-\alpha)v_{-1}.
    \end{align*}
    By Lemma~\ref{lemma:value_game_minus_1}, we have $v_{-1}>v$. Therefore the above inequality is impossible unless $\alpha=1$. Thus $[\bfy^*]_1=1$, and since $\bfy^* \in S_m$, it follows that $\bfy^*=\bfe_1$.
\end{pf}

\begin{pf}[of Lemma~\ref{lemma:one_component_less_v}]
Let $\bfx \in \partial S_n$. Since Assumption~\ref{asspt:boundary} implies $\partial S_n \cap \nex=\emptyset$, the point $\bfx$ is not a solution of the max-min problem~\eqref{eq:lp_duality_x}. Hence
\begin{align*}
    \min\{\bfc_1^\top \bfx,\dots,\bfc_m^\top \bfx\}<v.
\end{align*}
Because $\bfc_1=v\onevector_n$, we have $\bfc_1^\top \bfx=v$, so there must exist some $i\ge 2$ such that
\begin{align*}
    \bfc_i^\top \bfx < v.
\end{align*}

Now suppose, for contradiction, that a column $\bfc_j$ has no entry strictly less than $v$. Then every entry of $\bfc_j$ is at least $v$. By Assumption~\ref{asspt:deduced_from_A1}, we have $\bfc_j\neq v\onevector_n$, so in fact $\bfc_j$ has at least one entry strictly greater than $v$. Therefore, for every $\bfx \in \intr(S_n)$,
\begin{align*}
    \bfc_j^\top \bfx > v = \bfc_1^\top \bfx.
\end{align*}
On the other hand, for every $\bfx \in \partial S_n$, the first part of the proof shows that there exists some $i\ge 2$ such that
\begin{align*}
    \bfc_i^\top \bfx < v \le \bfc_j^\top \bfx.
\end{align*}
Hence column $j$ is never a best response of Player~2 to any $\bfx \in S_n$. Since an action that is never a best response is strictly dominated by some mixed action~\citep{pearce1984rationalizable}, Action~$j$ is strictly dominated for Player~2, contradicting Assumption~\ref{asspt:no_dominated}. Therefore every column of $A$ must have at least one entry strictly less than $v$.
\end{pf}

\begin{pf}[of Lemma~\ref{lemma:binding}]
Since $\bfx^* \in \nex$, we have $\min_{l\in[m]}\bfc_l^\top \bfx^* = v$, and $\bfc_1^\top\bfx^*=v$ gives $1\in I$. Suppose, for contradiction, that $\bfc_i^\top\bfx^*>v$ for every $i\ge2$. By continuity there is a neighborhood $U$ of $\bfx^*$ in $S_n$ on which $\bfc_i^\top\bfx>v$ for all $i\ge2$; since $\bfc_1^\top\bfx=v$ identically, $U\subseteq\nex$, so $\bfx^*\in\intr(\nex)$, contradicting $\bfx^*\in\partial\nex$. Hence some $i\ge 2$ satisfies $\bfc_i^\top\bfx^*=v$, i.e., $i\in I$.
\end{pf}

\bibliographystyle{cas-model2-names}
\bibliography{bib}

\end{document}